\documentclass[sn-mathphys,Numbered]{sn-jnl}
\usepackage{graphicx}
\usepackage{multirow}
\usepackage{amsmath,amssymb,amsfonts}
\usepackage{amsthm}
\usepackage{mathrsfs}
\usepackage[title]{appendix}
\usepackage{xcolor}
\usepackage{textcomp}
\usepackage{manyfoot}
\usepackage{booktabs}
\usepackage{algorithm}
\usepackage{algorithmicx}
\usepackage{algpseudocode}
\usepackage{listings}
\theoremstyle{thmstyleone}
\newtheorem{theorem}{Theorem}
\newtheorem{proposition}[theorem]{Proposition}
\theoremstyle{thmstyletwo}
\newtheorem{example}{Example}
\newtheorem{remark}{Remark}

\theoremstyle{thmstylethree}
\newtheorem{definition}{Definition}

\newtheorem{lemma}[theorem]{Lemma}
\begin{document}

\title[Two RSA-based Cryptosystems ]{Two RSA-based Cryptosystems }

\author*[1]{\fnm{A.} \sur{Telveenus}}\email{t.fernandezantony@kingston.ac.uk}

\affil*[1]{\orgdiv{International Study Centre}, \orgname{Kingston University}, \orgaddress{\street{Kingston Hill Campus}, \city{London}, \postcode{KT2 7LB},  \country{UK}}}

\abstract{The cryptosystem RSA is a very popular cryptosystem in the study of Cryptography. In this article, we explore how the idea of a primitive $m^{th}$ root of unity in a ring can be integrated into the Discrete Fourier Transform, leading to the development of new cryptosystems known as RSA-DFT and RSA-HGR.}

\keywords{Primitive $m^{th}$ roots of unity,  halidon rings, halidon group rings, RSA cryptosystem,  Discrete Fourier Transform.}

\pacs[MSC Classification]{16S34, 20C05, 11T71}

\maketitle

\section{Introduction}
The development of asymmetric cryptography, also known as public-key cryptography, is the greatest and perhaps the only true revolution in the history of cryptography.A significant portion of the theory behind public-key cryptosystems relies on number theory \cite{ws}. RSA is a widely used public-key cryptosystem for secure data transmission. It was first described publicly in 1977 by Ron Rivest, Adi Shamir, and Leonard Adleman. An equivalent system had been secretly developed in 1973 by Clifford Cocks at the British signals intelligence agency, GCHQ. \\
In RSA, the encryption key is public, while the decryption key is private. To create a public key, a user chooses two large prime numbers and an auxiliary value, and then publishes this information. The prime numbers are kept secret. Anyone can use the public key to encrypt messages, but only someone who knows the private key can decrypt them.
The security of RSA relies on the factoring problem, which is the practical difficulty of factoring the product of two large prime numbers. Breaking RSA encryption is known as the RSA problem. RSA is a relatively slow algorithm, so it is not commonly used to directly encrypt user data. Instead, it is used to transmit shared keys for symmetric-key cryptography. The author establishes a connection between RSA and Halidon rings using Discrete Fourier Transform and Halidon group rings in this article.

In 1940, the famous celebrated mathematician Graham Higman published a theorem \cite{gh},\cite{gk} in group algebra which is valid only for a field or an integral domain with some specific conditions. In 1999, the author noticed that this theorem can be extended to a rich class of rings called halidon rings\cite{at}.  \\

A primitive $m^{th}$ root of unity in a ring with unit element is completely different from that of in a field, because of the presence of nonzero zero divisors.  So we need a separate definition for a primitive $m^{th}$ root of unity. An element $\omega $ in a ring $R$ is called a  \textit{primitive} $m^{th}$ root if $m$ is the least positive integer such that  $\omega^{m}=1$ and
\begin{eqnarray*}
\sum_{r=0}^{m-1} \omega^{r(i-j)}&=& m, \quad  i= j (\ mod \ m )\\ &=& 0, \quad  i\neq j (\ mod \ m ). \end{eqnarray*} \\
More explicitly,
\begin{eqnarray*}
1+ \omega^{r}+(\omega^{r})^{2}+(\omega^{r})^{3}+(\omega^{r})^{4}+......+(\omega^{r})^{m-1}&=& m, \quad  r=0 \\ &=& 0, \quad 0<r\leq m-1. \end{eqnarray*} \\
A ring $R$ with unity is called a \textit{halidon} ring with index $m$ if there is a  primitive $m^{th}$ root of unity and $m$ is invertible in $R$. The ring of integers is a halidon ring with index $ m=1$ and $\omega=1$. The halidon ring with index $1$ is usually called a \textit{trivial} halidon ring. The field of real numbers is a halidon ring with index $m=2$ and $\omega = -1$. The field $\mathbb{Q}$ $(i)=\{ a+ib | a,b \in$ $\mathbb{Q}$ $\}$ is a halidon ring with $\omega=i$ and $m=4$.  $\mathbb{Z}_{p}$ is a halidon ring with index $p-1$ for every prime $p$. Interestingly, $\mathbb{Z}_{p^{k}}$ is also a halidon ring with same index for any integer $k>0$ and it is not a field if $k>1.$ Note that if $\omega$ is a primitive $m^{th}$ root of unity, then $\omega^{-1}$ is also a primitive $m^{th}$ root of unity.
\section{Preliminary results}
In this section, we state some new results and the results essential for constructing the RSA-DFT and RSA-HGR Cryptosystems only. The readers who are interested in the properties of halidon ring, can refer to \cite{at}, \cite{ath} and \cite{ath1}. \\
Let $U(R)$ denote the unit group of $R$.The following theorem will give the necessary and sufficient conditions for a ring to be a halidon ring. The author has used this theorem to develop the computer programme-1 given below.
\begin{theorem} \label{rd2} (A. Telveenus \cite{ath})
A finite commutative ring $R$ with unity is a halidon ring with index $m$ if and only if there is a primitive $m^{th}$ root of unity  $\omega$ such that $m$, $\omega^{d}-1 \in U(R)$; the unit group of $R$ for all divisors $d$ of $m$ and $d<m$.
\end{theorem}
\begin{proposition} \label{rd16}
The homomorphic image of a commutative halidon ring with index $m$ and primitive $m^{th}$ of unity $\omega$ is also a halidon ring with index $m$.
\end{proposition}
\begin{proposition}\label{rd20}
Let R be a commutative halidon ring with index m and let k $>$ 1 be a divisor of m. Then R is also a halidon ring with index k.
\end{proposition}
In the rest of the section, let $R=\mathbb{Z}_{n}$ be a halidon ring with index $m$ and primitive $m^{th}$ root of unity $\omega$.
\begin{lemma} \label{rd3}
Let $p$ be an odd prime number and $k$ a positive integer. Then
\begin{enumerate}
  \item $U(\mathbb{Z}_{p})=<\omega>$ for some $\omega \in U(\mathbb{Z}_{p})$ with order $p-1$,
  \item $U(\mathbb{Z}_{p^{k}})=<\omega>$ for the same $\omega$ treating as an element in  $U(\mathbb{Z}_{p^{k}})$ with order $\phi(p^{k})$.
\end{enumerate}
\end{lemma}
\begin{proof}
\begin{enumerate}
  \item Since $\mathbb{Z}_{p}$ is a field, the result follows.
  \item  $\omega^{p-1}=1 \ mod \ p \Rightarrow \omega^{p-1}=1+lp$, where $l$ is an integer and using the binomial expansion, we get $\omega^{p^{k-1}(p-1)}=1 \ mod \ p^{k}$ . Let $s$ be the order of $\omega$ in $U(\mathbb{Z}_{p^{k}})$. Therefore $s \leq p^{k-1}(p-1)$. If $s< p^{k-1}(p-1)$, then $s|p^{k-1}(p-1)$ and this implies $s|p-1$ which is not possible as the order of $\omega $ in $U(\mathbb{Z}_{p})$ is $p-1$. Thus we have the order of $\omega$ in $U(\mathbb{Z}_{p^{k}})$ as $p^{k-1}(p-1)=\phi(p^{k})$.
\end{enumerate}
\end{proof}
\begin{definition}
Let R be a ring and $\alpha \in R$. The element $\alpha$ is called a \textit{primitive element} or \textit{primitive root} if $\alpha$ multiplicatively generates the unit group $U(R)$
 of the ring R.
\end{definition}
\begin{example}
$U(\mathbb{Z}_{5})=<2>$ with order 4  and $U(\mathbb{Z}_{5^{3}})=<2>$ with order $\phi(125)=100$. Clearly, $2 \in \mathbb{Z}_{5^{3}}$ is a primitive root but not a primitive root of unity in $\mathbb{Z}_{5^{3}}$.
\end{example}
\begin{proposition}  \label{rd17}
Let $p$ be an odd prime number. Then $\mathbb{Z}_{p^{k}}$ is a halidon ring with index $m=p-1$ and  $\omega_{1}=\omega^{p^{k-1}}$ is a primitive $m^{th}$ root of unity for positive integers $k\geq 1$.
\end{proposition}
\begin{proof}
Since $\mathbb{Z}_{p}$ is a halidon ring with  $\omega$ as a primitive $m^{th}$ root of unity, $\omega^{m}=\omega^{p-1}=1$ and $m, \ \omega^{r}-1 \in U(Z_{p})$ for $ r=1,2,3,..,m-1 $. Clearly $\omega_{1}^{m}=1$ and  $\omega_{1}^{r}-1=\omega^{p^{k-1}r} -1 \in U(Z_{p})$ for $r=1,2,3,..m-1$. By lemma \ref{rd3}, $\mathbb{Z}_{p^{k}}$ is a halidon ring with index $m=p-1$.
\end{proof}
The complete characterisation of the halidon property in $\mathbb{Z}_{n}$, where $n$ is odd, is given by the following theorem.
\begin{theorem}\label{rd18}
The ring $\mathbb{Z}_{n}$, where $ \displaystyle n=p_{1}^{e_{1}}p_{2}^{e_{2}}p_{3}^{e_{3}}.....p_{k}^{e_{k}}$ with $2<p_{1}<p_{2}<.....<p_{k}$ is a halidon ring with index $m$ and the primitive $m^{th}$ root of unity $\omega$ if and only if each $ \displaystyle \mathbb{Z}_{p_{i}^{e_{i}}}$ is a halidon ring with index $m$ and primitive $m^{th}$ root of unity $ \displaystyle \omega_{i}=\omega \ mod \ {p_{i}^{e_{i}}} $ for each $i=1,2,3,...,k$.
\end{theorem}
\begin{proof}
We define a map $ \displaystyle f:\mathbb{Z}_{n} \rightarrow \Pi_{i=1}^{k} \mathbb{Z}_{p_{i}^{e_{i}}}$ by $ \displaystyle f(a)=\Pi_{i=1}^{k}a \ mod \ p_{i}^{e_{i}}$. Clearly $f$ is an isomorphism and using the proposition \ref{rd16}, $ \displaystyle \omega \in U(\mathbb{Z}_{n})$ is a primitive $m^{th}$ root of unity if and only if each $ \displaystyle \omega_{i}=\omega \ mod \ p_{i}^{e_{i}} \in U(\mathbb{Z}_{p_{i}^{e_{i}}})$ is a primitive $m^{th}$ root of unity in $\mathbb{Z}_{p_{i}^{e_{i}}}$. If $\omega_{i}$'s are known, then we can calculate $ \displaystyle \omega \in U(\mathbb{Z}_{n})$ using the Chinese remainder theorem.
\end{proof}
\begin{definition} \label{rd7}
Let $p_{1},p_{2},p_{3},....,p_{k}$ be odd primes and let $\phi(x)$ be the Euler's totient function. We define the \textit{halidon function}  $$\psi(n)= \begin{cases} gcd \{ \phi(p_{1}^{e_{1}}), \phi(p_{2}^{e_{2}}),\phi(p_{3}^{e_{3}}),...., \phi(p_{k}^{e_{k}})\}, & n=p_{1}^{e_{1}}p_{2}^{e_{2}}p_{3}^{e_{3}}.....p_{k}^{e_{k}} \\
1, & n \ \text{is even}\end{cases} $$
\end{definition}
\begin{proposition}
Let $n$ be as in definition \ref{rd7}. Then the halidon function $$\psi(n)=gcd\{ p_{1}-1,p_{2}-1,p_{3}-1,....,p_{k}-1\},$$ which is independent of the exponents $e_{1},e_{2},e_{3},....,e_{k}$.
\end{proposition}
Now we can prove the following theorem, which was previously a conjecture(see \cite{ath1}):
 \begin{theorem}\label{rd1}(A. Telveenus \cite{ath1})
If $R=\mathbb{Z}_{n}$ and $n=p_{1}^{e_{1}}p_{2}^{e_{2}}p_{3}^{e_{3}}.....p_{k}^{e_{k}}$ with primes $p_{1}<p_{2}<p_{3}<....<p_{k}$ including 2, then $R$ is a halidon ring with maximal index $m_{max}=\psi(n)$.
\end{theorem}
\begin{proof}
Suppose $n$ is odd. Since the map $ \displaystyle f:\mathbb{Z}_{n} \rightarrow \Pi_{i=1}^{k} \mathbb{Z}_{p_{i}^{e_{i}}}$ by $ \displaystyle f(a)=\Pi_{i=1}^{k}a \ mod \ p_{i}^{e_{i}}$ is an isomorphism, the order of $\omega$, $o(\omega)=m \Leftrightarrow o(\omega_{i})=m$, for $i=1,2,3,..,k.$ Also, $\omega_{i}\in \mathbb{Z}_{p_{i}^{e_{i}}}\Rightarrow \omega_{i}^{\phi(p_{i}^{e_{i}})}=1 \Rightarrow m \mid \phi(p_{i}^{e_{i}})=p_{i}^{e_{i}-1}(p_{i}-1) $. Since $m$ is even, $m \nmid p_{i}^{e_{i}-1} \Rightarrow m \mid p_{i} -1 $ for $i=1,2,3,..,k \Rightarrow  m_{max}=gcd(p_{1}-1, p_{2}-1, p_{3}-1, ....,p_{k}-1)=\psi(n)$ by definition \ref{rd7}. If $n$ is even, then $m=1=\psi(n)$ as $m$ is invertible in $R$.
\end{proof}
From proposition \ref{rd17}, we have $\mathbb{Z}_{p^{k}}$ is a halidon ring with index $m=p-1$ and the primitive $m^{th}$ root of unity $\omega_{1}=\omega^{p^{k-1}}$ for positive integer $k\geq 1$ where $\omega$ is a primitive $m^{th}$ root of unity in $\mathbb{Z}_{p}$.
\begin{lemma} \label{rd11}
Let $p$ be a prime number. Then the number of primitive $k^{th}$ roots of unity in $\mathbb{Z}_{p}$ is $\phi(k)$.
\end{lemma}
\begin{proof}
Since $\mathbb{Z}_{p}$ is a field, it is a halidon ring with maximum index $p-1$. It is clear that every non zero element in $\mathbb{Z}_{p}$ is a primitive $k^{th}$ root of unity for some positive integer $k| p-1$. It is well known that $$\sum _{d|p-1}\phi(d)=p-1.$$ $\therefore \quad $ the number of  $k^{th}$ root of unity in $\mathbb{Z}_{p}$ is $\phi(k).$
\end{proof}
\begin{theorem} \label{rd19}
Let $n=p_{1}^{e_{1}}p_{2}^{e_{2}}p_{3}^{e_{3}}.....p_{k}^{e_{k}}$ such that $e_{i}>0$ are integers and $p_{i}=mt_{i}+1;$ where $m$ is the maximum index of the halidon ring $\mathbb{Z}_{n}$ and $t_{i}$'s are relatively prime for all $i=1,2,3...,k$. Then the number of primitive $m^{th}$ root of unity in $\mathbb{Z}_{n}$ is $[\phi(m)]^{k}$.
\end{theorem}
\begin{proof}
The result follows from theorem \ref{rd18}, theorem \ref{rd1} and lemma \ref{rd11}.
\end{proof}
The following computer code based on theorem \ref{rd2} computes primitive $m^{th}$ root of unity in the halidon ring $\mathbb{Z}_{n}$ and verifies theorem \ref{rd19}.  
\begin{verbatim}
Programme-1: To check whether Z(n) is a trivial or nontrivial halidon ring.
#include <iostream>
#include <cmath>
using namespace std;
int main() {
	cout << "To check whether Z(n) is a trivial or nontrivial halidon ring." << endl;
	unsigned long long int t = 0, n = 1, w = 1, hcf, hcf1,
		d = 1, k = 1, q = 1, p = 1, b=0, c=0, temp = 1;
	cout << "Enter an integer n >0: ";
	cin >> n;
	if (n % 2 == 0) {
		cout << "Z(" << n << ") is a trivial halidon ring." << endl;
	}
		for (w = 1; w < n; ++w) {
		for (int i = 1; i <= n; ++i) {
			if (w % i == 0 && n % i == 0) {
				hcf = i;

		}
		} if (hcf == 1) {
			++t; // cout << "  " << w << "  ";
		}
	}
	for (w = 1; w < n; ++w) {
		for (int i = 1; i <= n; ++i) {
			if (w % i == 0 && n % i == 0) {
				hcf = i;

			}
		}
		if (hcf == 1) {
		for (int k = 1; k <= t; ++k) {
		q = q * w; q = q % n;
		if (q == 1) { if (temp <= k) { temp = k; } break; }
		}
		}
	}
	
	for (w = 2; w < n; ++w) {
		for (int i = 1; i <= n; ++i) {
		if (w % i == 0 && n % i == 0) {
		hcf = i;
			}
		}
		if (hcf == 1) {
		for (int k = 1; k <= temp; ++k) {
		q = q * w; q = q % n; if (q == 1) {
		for (int i = 1; i <= n; ++i) {
		if (k % i == 0 && n % i == 0) {
		hcf1 = i;
		}
		}
		if (hcf1 == 1) {
		for (int j = 1; j < k; ++j)
		{
		if (k%j == 0) {
		d = j;
		for (int l = 1; l <= d; ++l)
		{
		p = (p*w); p = p % n;
		}
		for (int i = 1; i <= n; ++i) {
		if ((p - 1) % i == 0 && n % i == 0) {
		hcf = i;
		}
		}
             if (hcf == 1) {

		p = 1; b = b + 1;
		}
		else p = 1;
		c = c + 1;
     		}
		}
		if (c == b) { cout << "   Z(" << n << ")" <<
							
" is a halidon ring with index m= " << k <<

" and w= " << w << "."<< endl; } {p = 1; c = 0; b = 0; }

		break;
		}
		}
		}
		}
	} return 0;
}
\end{verbatim}

 \section{Discrete Fourier Transforms}
In this section, we deal with the ring of polynomials over a halidon ring which has an application in Discrete Fourier Transforms \cite{jj}. Throughout this section, let $R$ be a finite commutative halidon ring with index $m$ and $R[x]$ denotes the ring of polynomials degree less than $m$ over $R$. Also, refer to \cite{ath1}.
\begin{definition} \cite{jj}
Let $\omega \in R$ be a primitive $m^{th}$ root of unity in $R$ and let $f(x)=\sum \limits_{j=0}^{m-1} f_{j}x^{j} \in R[x]$ with its coefficients vector $(f_{0},f_{1},f_{2},....,f_{m-1}) \in R^{m}$. The \textbf{Discrete Fourier Transform} (DFT) is a map $$ DFT_{\omega}: R[x]\rightarrow R^{m}$$ defined by $$DFT_{\omega}(f(x))=(f_{0}(1),f_{1}(\omega),f_{2}(\omega^{2}),....,f_{m-1}(\omega^{m-1})).$$
\end{definition}\noindent
\begin{remark}
Clearly $DFT_{\omega}$ is a $R$-linear map as $DFT_{\omega}(af(x)+bg(x))=aDFT_{\omega}(f(x))+bDFT_{\omega}(g(x))$ for all $ a, b \in R$. Also, if $R= \mathbb{C}$, the field of complex numbers, then $\omega=cos(\frac{2 \pi}{m})+i sin(\frac{2 \pi}{m})=e^{i\frac{2 \pi}{m}}$ and the Fourier series will become the ordinary series of sin and cos functions.
\end{remark}
\begin{definition} \cite{jj}
The \textbf{convolution} of $f(x)=\sum \limits_{j=0}^{m-1} f_{j}x^{j}$ and $g(x)=\sum \limits_{k=0}^{m-1} g_{k}x^{k}$ in $R[x]$ is defined by $h(x)=f(x)*g(x)= \sum \limits_{l=0}^{m-1} h_{l}x^{l} \in R[x]$ where  \quad $h_{l}=\sum \limits_{j+k=l \ mod \ m} f_{j}g_{k}=\sum \limits_{j=0}^{m-1} f_{j}g_{l-j}$ for $0 \leq l < m$.
\end{definition}
The notion of convolution is equivalent to polynomial multiples in the ring $R[x]/<x^{m}-1>$. The $l^{th}$ coefficient of the product $f(x)g(x)$ is $\sum \limits_{j+k=l \ mod \ m} f_{j}g_{k}$ and hence $$f(x)*g(x)= f(x)g(x) \ mod (x^{m}-1).$$
\begin{proposition} \cite{jj}
For polynomials $f(x),g(x) \in R[x],$ $DFT_{\omega}(f(x)*g(x))=DFT_{\omega}(f(x)).DFT_{\omega}(g(x)),$ where . denotes the pointwise multiplication of vectors.
\end{proposition}
\begin{proof}
$f(x)*g(x)= f(x)g(x) +q(x)(x^{m}-1)$ for some $q(x) \in R[x]$. \\
Replace $x$ by $\omega^{j}$, we get \\
$$f(\omega^{j})*g(\omega^{j})= f(\omega^{j})g(\omega^{j}) +0.$$
$$\therefore \quad \quad \quad DFT_{\omega}(f(x)*g(x))=DFT_{\omega}(f(x)).DFT_{\omega}(g(x)). $$
\end{proof}
\begin{theorem} \label{rd9}
For a polynomial $f(x) \in R[x],$ $DFT_{\omega}^{-1}(f(x))= \frac{1}{m}DFT_{\omega^{-1}}(f(x)).$
\end{theorem}
\begin{proof}
The matrix of the transformation $DFT_{\omega}(f(x))$ is $$[DFT_{\omega}(f(x))]=\phi=\left(
                                                                        \begin{array}{ccccc}
                                                                          1 & 1 & 1 & ..... & 1 \\
                                                                          1 & \omega & \omega^{2} & ..... & \omega^{m-1} \\
                                                                          1 & \omega^{2} &(\omega^{2})^{2} & ..... & (\omega^{2})^{m-1}\\
                                                                          . & . & . & ..... & . \\
                                                                          . & . & . & ..... & . \\
                                                                          . & . & . & ..... & . \\
                                                                          1 & \omega^{m-1} & (\omega^{m-1})^{2} & ..... & (\omega^{m-1})^{m-1}\\
                                                                        \end{array}
                                                                      \right)
$$ The matrix $\phi$ is the well known Vandermonde matrix and its inverse is $\frac{1}{m}\phi^{*}$, where $\phi^{*}$ is the matrix transpose conjugated \cite{pjd}. Since $\phi$ is a square matrix and the conjugate of $\omega$ is $\omega^{-1}$, we have  $DFT_{\omega}^{-1}(f(x))= \frac{1}{m}DFT_{\omega^{-1}}(f(x)).$
\end{proof}
\begin{example} \label{rd6}
We know that $R=Z_{49}$ is a halidon ring with index $m=6$ and $\omega=19$. Also, $\omega^{-1}=\omega^{5}=31$. Let $f(x)=2+x+2x^{2}+3x^{3}+5x^{4}+10x^{5} \in R[x].$ Then $DFT_{\omega}(f(x))$ can be expressed as \newline $ \left(
               \begin{array}{c}
                 F_{0} \\
                 F_{1} \\
                 F_{2} \\
                 F_{3} \\
                 F_{4} \\
                 F_{5} \\
               \end{array}
             \right)$ $=$ $\left(
                         \begin{array}{cccccc}
                           1 & 1 & 1 & 1 &1 & 1 \\
                           1 & \omega & \omega^{2} &\omega^{3} & \omega^{4} & \omega^{5} \\
                           1 & \omega^{2} & \omega^{4} &1& \omega^{2} & \omega^{4} \\
                           1 & \omega^{3} & 1 &\omega^{3} & 1 & \omega^{3} \\
                           1 & \omega^{4} & \omega^{2} &1 & \omega^{4} & \omega^{2} \\
                           1 & \omega^{5} & \omega^{4} &\omega^{3} & \omega^{2} & \omega \\
                         \end{array}
                       \right)$  $ \left(
               \begin{array}{c}
                 f_{0} \\
                 f_{1} \\
                 f_{2} \\
                 f_{3} \\
                 f_{4} \\
                 f_{5} \\
               \end{array}
             \right)$
     $\Rightarrow$
     $ \left(
               \begin{array}{c}
                 F_{0} \\
                 F_{1} \\
                 F_{2} \\
                 F_{3} \\
                 F_{4} \\
                 F_{5} \\
               \end{array}
             \right)$ $=$ $ \left(
               \begin{array}{c}
                23  \\
                24 \\
                32  \\
                 44 \\
                 9 \\
                 27 \\
               \end{array}
             \right)$
       $ \left(
               \begin{array}{c}
                 f_{0} \\
                 f_{1} \\
                 f_{2} \\
                 f_{3} \\
                 f_{4} \\
                 f_{5} \\
               \end{array}
             \right)$ $=6^{-1}$$\left(
                         \begin{array}{cccccc}
                           1 & 1 & 1 & 1 &1 & 1 \\
                           1 & \omega^{5} & \omega^{4} &\omega^{3} & \omega^{2} & \omega \\
                           1 & \omega^{4} & \omega^{2} &1& \omega^{4} & \omega^{2} \\
                           1 & \omega^{3} & 1 &\omega^{3} & 1 & \omega^{3} \\
                           1 & \omega^{2} & \omega^{4} &1 & \omega^{2} & \omega^{4} \\
                           1 & \omega & \omega^{2} &\omega^{3} & \omega^{4} & \omega^{5} \\
                         \end{array}
                       \right)$  $ \left(
               \begin{array}{c}
                 F_{0} \\
                 F_{1} \\
                 F_{2} \\
                 F_{3} \\
                 F_{4} \\
                 F_{5} \\
               \end{array}
             \right)$
             \newline
      $\Rightarrow$     $ \left(
               \begin{array}{c}
                 f_{0} \\
                 f_{1} \\
                 f_{2} \\
                 f_{3} \\
                 f_{4} \\
                 f_{5} \\
               \end{array}
             \right)$ $=41$$\left(
                         \begin{array}{cccccc}
                           1 & 1 & 1 & 1 &1 & 1 \\
                           1 & 31 & 30 &48 & 18 & 19 \\
                           1 & 30 & 18 &1& 30 & 18 \\
                           1 & 48 & 1 &48 & 1 & 48 \\
                           1 & 18 & 30 &1 & 18 & 30 \\
                           1 & 19& 18 &48 & 30 & 31 \\
                         \end{array}
                       \right)$  $ \left(
               \begin{array}{c}
                 F_{0} \\
                 F_{1} \\
                 F_{2} \\
                 F_{3} \\
                 F_{4} \\
                 F_{5} \\
               \end{array}
             \right)$     \newline
If $ \left(
               \begin{array}{c}
                 F_{0} \\
                 F_{1} \\
                 F_{2} \\
                 F_{3} \\
                 F_{4} \\
                 F_{5} \\
               \end{array}
             \right)$     $=$ $ \left(
               \begin{array}{c}
                 23 \\
                 24 \\
                 32 \\
                 44 \\
                 9 \\
                 27 \\
               \end{array}
             \right)$, then a direct calculation gives $ \left(
               \begin{array}{c}
                 f_{0} \\
                 f_{1} \\
                 f_{2} \\
                 f_{3} \\
                 f_{4} \\
                 f_{5} \\
               \end{array}
             \right)$ $=$  $ \left(
               \begin{array}{c}
                2 \\
                1 \\
                 2 \\
                 3 \\
                 5 \\
                 10 \\
               \end{array}
             \right)$ \newline as expected.

\end{example}
Programmes 2 and 3  will enable us to calculate Discrete Fourier Transform and its inverse. We can cross-check the programmes against example \ref{rd6}. \\
\begin{verbatim}
Programme-2: Discrete Fourier Transform

#include <iostream>
#include<cmath>
using namespace std;
int main()
{
cout << "Discrete Fourier Transform" << endl;
unsigned long long int a[1000][1000], b[1000][1000],
mult[1000][1000],	q=1,m=1, n=1, w2=1,w=1, r1, c1, r2,
c2, i, j, k, t=1;
cout << "Enter n,m,w: ";
	cin >> n >> m >> w;
	r1 = m; c1=m;
	r2 = m; c2=1;
	for (i = 0; i < r1; ++i)
	for (j = 0; j < c1; ++j)
		{
			 t = (i*j)%m;
	 if (t == 0) a[i][j] = 1;
	 else
	for (q = 1; q < t + 1; ++q) { w2 = (w2 * w) % n; }
			 a[i][j] = w2; w2 = 1;
		}
	for (i = 0; i < r1; ++i)
		for (j = 0; j < c1; ++j)
	   {
		cout<<"  a"<<i+1<<" "<<j+1<<"="<<	a[i][j] ;
		if (j == c1 - 1)
			cout << endl;
		}
	cout << endl << "Enter coefficient vector of
    the polynomial:" << endl;
	for (i = 0; i < r2; ++i)
	for (j = 0; j < c2; ++j)
		{
    cout << "Enter element f" << i << " = ";
	        cin >> b[i][j];
		}
	for (i = 0; i < r1; ++i)
	for (j = 0; j < c2; ++j)
		{
    		mult[i][j] = 0;
		}
	for (i = 0; i < r1; ++i)
	for (j = 0; j < c2; ++j)
	for (k = 0; k < c1; ++k)
		{
    		mult[i][j] += (a[i][k]) * (b[k][j]);
		}
	cout << endl << "DFT Output: " << endl;
	for (i = 0; i < r1; ++i)
	for (j = 0; j < c2; ++j)
		{
			cout << "F"<< i << "="<< mult[i][j]%n;
	if (j == c2 - 1)
			cout << endl;
		}
	return 0;
        }


Programme-3: Inverse Discrete Fourier Transform

#include <iostream>
#include<cmath>
using namespace std;
int main()
{
	cout << "Inverse Discrete Fourier Transform" << endl;
	unsigned long long int a[1000][1000], b[1000][1000],
    mult[1000][1000], p=1, q=1, l=1, m = 1, m1 = 1, w1 = 1,
    w2=1, n = 1, w = 1, r1, c1, r2, c2, i, j, k,
    c=1,t = 1;
	cout << "Enter n,m, w: ";
	cin >> n >> m >> w;
	for (l = 1; l < n; ++l)
        {
		c = (l * m) % n;
		if (c == 1)
        {
			m1 = l;
		}
	    }
	for (p = 1; p < m; ++p)
    	{
		w1 = (w1 * w) % n;
	    }
	r1 = m; c1 = m;
	r2 = m; c2 = 1;
	for (i = 0; i < r1; ++i)
	for (j = 0; j < c1; ++j)
		{
			t = (i * j) % m;
	if (t == 0) a[i][j] = 1;
	else
	for (q = 1; q < t + 1; ++q) { w2 = (w2 * w1) % n; }
			a[i][j] = w2; w2 = 1;
		}
	for (i = 0; i < r1; ++i)
	for (j = 0; j < c1; ++j)
		{
	cout << "  a" << i + 1 << j + 1 << "=" << a[i][j];
	if (j == c1 - 1)
	cout << endl;
		}
	cout << endl << "Enter DFT vector :" << endl;
	for (i = 0; i < r2; ++i)
	for (j = 0; j < c2; ++j)
		{
	cout << "Enter element F" << i << " = ";
	cin >> b[i][j];
		}
	for (i = 0; i < r1; ++i)
	for (j = 0; j < c2; ++j)
		{
	mult[i][j] = 0;
		}
	for (i = 0; i < r1; ++i)
	for (j = 0; j < c2; ++j)
	for (k = 0; k < c1; ++k)
		{
	mult[i][j] += (a[i][k]) * (m1 * b[k][j]);
   		}
	cout << endl << "Polynomial vector: " << endl;
	for (i = 0; i < r1; ++i)
	for (j = 0; j < c2; ++j)
		{
	cout << "f" << i << "=" << mult[i][j] % n;
	if (j == c2 - 1)
	cout << endl;
		}
	return 0;
        }

\end{verbatim}

If $R=Z_{100001}$, $m=10$, $\omega=26364 $ and $f(x)=1+2x+3x^{2}+4x^{3}+5x^{4}+6x^{5}+7x^{6}+8x^{7}+9x^{8}+x^{9} \in R[x]$, then $ \left( \begin{array}{c}
                 F_{0} \\
                 F_{1} \\
                 F_{2} \\
                 F_{3} \\
                 F_{4} \\
                 F_{5} \\
                 F_{6} \\
                 F_{7} \\
                 F_{8} \\
                 F_{9} \\
               \end{array}
             \right)$     $=$ $ \left(
               \begin{array}{c}
                 46 \\
                 19019 \\
                 3314 \\
                 10082 \\
                 48017 \\
                 4 \\
                80347 \\
                 18172 \\
                 68413 \\
                 52627 \\
               \end{array}
             \right)$. \newline
Also, we can verify the inverse DFT using the above data. \\ The following proposition from number theory is very useful in the next section.
\begin{proposition} \cite{jl} \label{rd8}
Let $n=p_{i}^{e_{1}}p_{2}^{e_{2}}....p_{k}^{e_{k}}$ be the standard form of the integer $n$ and let d,e satisfy $ed\equiv 1 \ mod \ \phi(n)$. Then for all integer $x$,
$$x^{ed}\equiv x \ mod \ n.$$ Therefore, if $c=x^{e} \mod \ n$, we have $x \equiv c^{d}\ mod \ n. $
\end{proposition}
Let $$u= \sum_{i=1}^{m}\alpha_{i}g_{i}$$ be an element in the group algebra $RG$ and let  \begin{equation}\lambda_{r}=\sum_{i=1}^{m}\alpha_{m-i+2}(\omega^{(i-1)})^{(r-1)}\end{equation} where $\omega \in R$ is a primitive $m^{th}$ root of unity. Then $u$ is said to be \textit{depending } on $\lambda_{1},\lambda_{2},......,\lambda_{m}$.

\begin{theorem}\label{d1}
 Let $$u= \sum_{i=1}^{m}\alpha_{i}g_{i} \in U(RG)$$ be depending on $\lambda_{1},\lambda_{2},......,\lambda_{m}$. Let $$v=\sum_{i=1}^{m}\beta_{i}g_{i}$$
 be the multiplicative inverse of $u$ in $RG$. Then $$\beta_{i}=\frac{1}{m}\sum_{r=1}^{m} \lambda_{r}^{-1}(\omega^{i-1})^{r-1}.$$
 \end{theorem}
 Computer programme-4 \\
 \begin{verbatim}
#include<iostream>
#include<cmath>
using namespace std;
int main() {
	cout << "To find the units in Z(n)G;" <<
		"G is a cyclic group of order m through " <<
		"lamda take units in R." << endl;
long long int  a[1000], l[1000], w1[1000], m = 1, t = 0, x = 1, y = 1, s = 0,
		s1 = 0, m1 = 1, n = 1, i = 1, p=1,k = 0, r = 1, w = 1;

	cout << "Enter n =" << endl;
	cin >> n;
	cout << "Enter index m =" << endl;
	cin >> m;
	cout << "Enter m^(-1) =" << endl;
	cin >> m1;
	cout << "Enter primitive m th root w =" << endl;
	cin >> w;
	w1[0] = 1; cout << "w1[0]=" << w1[0] << endl;
	for (i = 1; i < m; ++i) {
	
		w1[i] = p * w % n; p = w1[i];
		cout << "w1[" << i << "]=" << w1[i] << endl;
	}
	cout << "Enter lamda values which are units" << endl;
	for (int i = 1; i < m + 1; ++i) {
		cout << "l[" << i << "]=" << endl;
		cin >> l[i];
	}
	for (int r = 1; r < m + 1; ++r) {
		for (int j = 1; j < m + 1; ++j)
		{
			x = ((j - 1) * (r - 1)) % m;
			k = k + (m1 * l[j] * w1[x]) % n; k = k % n;
			// cout << "k=" << k << endl;
		} a[r] = k; cout << "a[" << r << "]=" << a[r] << endl;
		k = 0;
	}
	cout << "The unit in RG is  u= ";
	s = 1;
mylabel:
	cout << a[s] << "g^(" << s - 1 << ") + ";
	s++;
	if (s < m + 1) goto mylabel; cout << endl;
	cout << endl;
	cout << "Note: Please neglect the last + as it is unavoidable for a for loop.";
	return 0;
}

 \end{verbatim}
 Computer Programme-5 \\
 
 \begin{verbatim}
#include<iostream>
#include<cmath>
using namespace std;
int main() {
	cout << "To check whether an element in Z(n)G;" <<
		"G is a cyclic group of order m" <<"has a multiplicative inverse or not" << endl;
long long	int a[1000], b[1000], c[1000], d[1000], e[1000], w1[1000], m = 1,
		t = 0, x = 1, s = 0, s1 = 0, l = 0, m1 = 1, hcf = 1,
		n = 1, i = 1, k = 0, q = 1, p = 1, r = 1, w = 1;
	cout << "Enter n =" << endl;
	cin >> n;
	cout << "Enter index m =" << endl;
	cin >> m;
	cout << "Enter  m^(-1) =" << endl;
	cin >> m1;
	cout << "Enter primitive m th root w =" << endl;
	cin >> w;
	w1[0] = 1;
	for (i = 1; i < m; ++i)
	{
		
		w1[i] = p * w % n; p = w1[i];
		cout << "w1[" << i << "]" << w1[i] << endl;
	}
	for (int i = 1; i < m + 1; ++i) {
		cout << "Enter a[" << i << "]=" << endl;
		cin >> a[i];
	}
	a[0] = a[m];
	for (int r = 1; r < m + 1; ++r) {
		for (int j = 1; j < m + 1; ++j)
		{
			l = (m - j + 2) % m;
			x = ((j - 1) * (r - 1)) % m;
			k = k + (a[l] * w1[x]) % n; k = k % n;
			// cout << "k=" << k << endl;
		} c[r] = k; cout << "lambda[" << r << "]=" << c[r] << endl;
		k = 0;
	}
	for (r = 1; r < m + 1; ++r) {
		for (int i = 1; i <= n; ++i) {
			if (c[r] % i == 0 && n % i == 0) {
				hcf = i;
			}
		}
		if (hcf == 1) {
			cout << "lambda[" << r << "] is a unit" << endl;
		}
		else {
			cout << "lambda[" << r <<
				"] is a not unit. So there is no multiplicative inverse." <<
				endl; t = 1;
		}
	}
	for (r = 1; r < m + 1; ++r) {
		for (int i = 1; i <= n; ++i) {
			e[r] = (c[r] * i) % n;
			if (e[r] == 1) {
				b[r] = i;
				cout << " The inverse of lambda[" << r << "] is " << b[r] << endl;
			}
		}
	}
	b[0] = b[m];
	for (int r = 1; r < m + 1; ++r) {
		for (int j = 1; j < m + 1; ++j)
		{
			l = (m - j + 2) % m;
			x = (m * m - (j - 1) * (r - 1)) % m;  //cout << "x= " << x << endl;
			k = k + (m1 * b[l] * w1[x]) % n; k = k % n;
			//cout << "k=" << k << endl;
		}
		d[r] = k; //cout << "d[" << r << "]=" << d[r] << endl;
		k = 0;
	}
	if (t == 1) {
		s = m;
	mylabel2:
		cout << a[m - s + 1] << "g^(" << m - s << ") + ";
		s--;
		if (s > 0) goto mylabel2; cout <<
			"has no multiplicative inverse." << endl;
	}
	else {
		cout << "The inverse of ";
		s = m;
	mylabel:
		cout << a[m - s + 1] << "g^(" << m - s << ") + ";
		s--;
		if (s > 0) goto mylabel; cout << "is" << endl;
		s1 = m;
	mylabel1:
		cout << d[m - s1 + 1] << "g^(" << m - s1 << ") + ";
		s1--;
		if (s1 > 0) goto mylabel1; cout << "." << endl;
	}
	return 0;
}

 \end{verbatim}
The computer programme-3 can be used to test whether a given element $u$ in $RG$ is a unit or not. If it is a unit, then the programme will give the multiplicative inverse $v$ in $RG$. \\
\begin{theorem}\label{rd4}
 Let $$u= \sum_{i=1}^{m}\alpha_{i}g_{i} \in RG$$ be depending on $\lambda_{1},\lambda_{2},......,\lambda_{m}$.
 Then \begin{enumerate} \label{rd9}
   \item $u \in U(RG)$ if and only if each $\lambda_{i} \in U(R)$,
  \item  $u \in E(RG)$ if and only if each $\lambda_{i} \in E(R)$, where $E(RG)$ is the set of idempotents in $RG$.
   \end{enumerate}
   More over, $|U(RG)|=|U(R)|^{|G|}$ and $|E(RG)|=|E(R)|^{|G|}$.
 \end{theorem}

\section{RSA-DFT Cryptosystem}
Let $m$ be the length of the message including the blank spaces between the words. If the message has a length more than $m$, we can split the message into blocks with lengths less than $m$. For a message of length less than $m$, we can add blank spaces after the period to make it a message with length $m$. Choose  large prime numbers such that $p_{i}=mt_{1}+1$ where $i=1,2,3,...k$ and  $t_{i}$'s are relatively prime. Let $\omega$ be a primitive $m^{th}$ root of unity in $\mathbb{Z}_{n}$, where $n=p_{1}^{e_{1}}p_{2}^{e_{2}}....p_{k}^{e_{k}}$ for some positive integers $e_{i}$. We know that $\mathbb{Z}_{n}$ is a halidon ring with maximum index $gcd(p_{1}-1,p_{2}-1,...p_{k}-1)$(theorem \ref{rd1}) and since $m|p_{1}-1,p_{2}-1,....p_{k}-1$, $\mathbb{Z}_{n}$ is also a halidon ring with index $m$(proposition \ref{rd20}).

Here we are considering a cryptosystem based on modulo $n$. The following table gives numbers and the corresponding symbols.
\begin{center}
\begin{tabular}{|c|c|}

 \hline
  Numbers assigned & Symbols \\
  \hline
  0 to 9 & 0 to 9 \\
  10 to 35  & A to Z \\
  36&blank space \\
  37 & colon \\
  38 & period \\
  39 & hyphen \\
\hline
\end{tabular}
\end{center}

In this cryptosystem, there are two stages. In stage 1, we shall compute the value of $\omega$ which Bob keeps secret and in stage 2, we shall decrypt the message sent by Bob. \\
\textbf{Stage 1-RSA}\\
\textbf{Cryptosystem setup}
\begin{enumerate}
  \item Alice chooses large primes $p_{1}, p_{2},...p_{k}$ and  positive integers $e_{1},e_{2},...e_{k}$, and calculates $n=p_{1}^{e_{1}}p_{2}^{e_{2}}...p_{k}^{e_{k}}$ and $\phi(n)=p_{1}^{e_{1}-1}p_{2}^{e_{2}-1}...p_{k}^{e_{k}-1}(p_{1}-1)(p_{2}-1)...(p_{k}-1)$.
  \item Alice chooses an $e$ so that $gcd(e,\phi(n))=1$.
  \item Alice calculates $d$ with property $ed\equiv \ 1 mod \ \phi(n)$.
  \item Alice makes $n$ and $e$ public and keeps the rest secret.
\end{enumerate}
\textbf{Cryptosystem Encryption}(Programme-1)
\begin{enumerate}
  \item Bob looks up Alice's $n$ and $e$ .
  \item Bob chooses an arbitrary $\omega\ mod \ n$ and kept secret.
  \item Bob sends $c \equiv \omega^{e} \ mod \ n$ to Alice.
\end{enumerate}
\textbf{Cryptosystem Decryption}(Proposition \ref{rd8})
\begin{enumerate}
  \item Alice receives $c$ from Bob.
  \item Alice computes $\omega\equiv \ c^{d}mod \ n$.
\end{enumerate}
\textbf{Stage 2-Discrete Fourier Transform}\\
\textbf{Cryptosystem setup}
\begin{enumerate}
  \item Alice chooses Discrete Fourier Transform as the encryption key.
  \item Alice chooses Inverse Discrete Fourier Transform as the decryption key.
\end{enumerate}
\textbf{Cryptosystem Encryption}(Programme-2)
\begin{enumerate}
  \item Bob looks up Alice's encryption key.
  \item Bob writes his message $x$.
  \item Bob computes $y=DFT \ x$ his chosen $\omega$.
  \item Bob sends $y$ to Alice.
\end{enumerate}

\textbf{Cryptosystem Decryption}(Programme-3)
\begin{enumerate}
  \item Alice receives $y$ from Bob.
  \item Alice computes $x=Inverse \ DFT \ y$ with $\omega$ calculated in stage 1.
\end{enumerate}
The above cryptosystem is an \textbf{asymmetric cryptosystem} as Alice and Bob share different information. For the practical application, we must choose very large prime numbers (more than 300 digits) so that the calculation of $\phi(n)$ must be very difficult and the probability of choosing the primitive $m^{th}$root of unity $\omega$  should tends to zero. We exhibit the working of the RSA-DFT cryptosystem using a simple choice of prime numbers in which the probability of choosing the primitive $m^{th}$root of unity $\omega$  is $\dfrac{1}{10000}=0.0001=0.01\%$ in the following example.
\begin{example}
\textbf{Stage 1}\\
The length of the message has been fixed as $m=202$.
Alice chooses two primes $p_{1}=607$and $p_{2}=809$ and two positive integers $e_{1}=1$ and $e_{2}=1$, and calculates $n=491063$ and $\phi(n)=489648$.\\
Alice chooses an $e=361123$ so that $gcd(e,\phi(n))=1$.\\
Alice calculates $d=18523$ with property $ed\equiv \ 1  \ mod \ \phi(n)$.
Alice shared Bob $n=491063$ and $e=361123$ and rest kept secret. \\
Bob looks up Alice's $n=491063$ and $e=361123$. \\
Bob chooses an arbitrary $202^{th}$ root of unity $\omega\ mod \ n$(there are $\phi(202)^{2}=10000$ $\omega$'s possible and they can be found by running the programme $1$ and it will take  around $8$ hours, and its probability is $\frac{1}{10000} = 0.0001)$ kept secret. \\
Bob sends $c \equiv \omega^{e} \ mod \ n \equiv 142638$ \ mod 491063 to Alice.
Alice receives $c$ from Bob.\\
Alice computes $\omega\equiv \ c^{d}mod \ n\equiv 239823 \ mod \ 491063$.\\
\textbf{Stage 2} \\

Alice shared Bob the encryption key Discrete Fourier Transform and  $n=491063.$
Suppose Bob sends the following secret message to Alice. \\
\begin{center} MY BANK DETAILS: NAME: JACK CARD NUMBER: 4125678 SORT CODE:20-30-41 ACCOUNT NUMBER:20164 BANK:OVERSEAS. \end{center}
The length of the message is 101 and to make it 202 we need to add 101 blank spaces. This can be translated into a 202 component vector \\

 $ \left(
                                   \begin{array}{c}
  22 \ 34 \ 36 \ 11 \ 10 \ 23 \ 20 \ 36 \ 13 \ 14 \ 29 \ 10 \ 18 \ 21 \ 28 \ 37 \\
  36 \ 23 \ 10 \ 22 \ 14 \ 37 \ 36 \ 19 \ 10 \ 12 \ 20 \ 36 \ 12 \ 10 \ 27 \ 13 \\
  36 \ 23 \ 30 \ 22 \ 11 \ 14 \ 27 \ 37 \ 4 \ 1 \ 2 \ 5 \ 6 \ 7 \ 8 \\
  36 \ 28 \ 24 \ 27 \ 29  \ 36 \ 12 \ 24 \ 13 \ 14 \ 37 \ 2 \ 0 \ 39 \ 3 \ 0 \ 39 \ 4 \ 1 \\
  36 \ 10 \ 12 \ 12 \ 24 \ 30 \ 23 \ 29 \ 23 \ 30 \ 22 \ 11 \ 14 \ 27 \ 37 \ 2\ 0 \ 1 \ 6 \ 4 \\
   36 \ 11\ 10 \ 23 \ 20 \ 37  \ 24 \ 31 \ 14 \ 27 \ 28 \ 14 \ 10 \ 28 \ 38 \\
  36 \ 36 \ 36 \ 36 \ 36 \ 36 \ 36 \ 36 \ 36 \ 36 \ 36 \ 36 \ 36 \ 36 \ 36 \ 36 \ 36 \ 36 \ 36 \ 36 \\
  36 \ 36 \ 36 \ 36 \ 36 \ 36 \ 36 \ 36 \ 36 \ 36 \ 36 \ 36 \ 36 \ 36 \ 36 \ 36 \ 36 \ 36 \ 36 \ 36 \\
  36 \ 36 \ 36 \ 36 \ 36 \ 36 \ 36 \ 36 \ 36 \ 36 \ 36 \ 36 \ 36 \ 36 \ 36 \ 36 \ 36 \ 36 \ 36 \ 36 \\
  36 \ 36 \ 36 \ 36 \ 36 \ 36 \ 36 \ 36 \ 36 \ 36 \ 36 \ 36 \ 36 \ 36 \ 36 \ 36 \ 36 \ 36 \ 36 \ 36 \\
  36 \ 36 \ 36 \ 36 \ 36 \ 36 \ 36 \ 36 \ 36 \ 36 \ 36 \ 36 \ 36 \ 36 \ 36 \ 36 \ 36 \ 36 \ 36 \ 36 \ 36 \ 36

                                   \end{array}
                                 \right)$

\noindent
Bob chooses a primitive $202^{th}$ roots of unity $\omega$ in stage 1 using  programme 1 and kept secret in the halidon ring $\mathbb{Z}_{491063}$ with index 202.
Applying DFT (see programme 2) to the plain text to get the following cipher text using the chosen value of $\omega$:  \\

$ \left(
                                   \begin{array}{c}
 5640 \ 28875 \ 82477 \ 377806 \ 380572 \ 399487 \ 350120 \ 214346 \\ 101686 \  277011 \ 93173 \ 220930 \ 72 573 \ 42514 \ 361289 \\
 476177 \ 371780 \ 243907 \ 179047 \ 292166 \ 427665 \ 243623 \ 344397 \\ 155022 \ 360049 \ 312478 \ 305875 \ 392901 \ 193460 \ 440042 \\
 ....................................................................\\..................................................................\\
 .....................................................................\\..................................................................\\

 216155 \ 440701 \ 157904 \ 342869 \ 348795 \ 159340 \ 140193 \ 222089 \\ 326519 \ 95581 \ 431250 \ 15009 \ 166938 \ 384271 \ 452109
                                     \end{array}
                                 \right) $

\noindent

The readers can check the above results by copying the programmes and paste in Visual Studio 2022 c++ projects. \\

Alice receives the above cipher text and  she uses  $\omega=239823$ from stage 1. Applying the inverse DFT (see programme 3)Alice gets the original message back. Also, we can assign letters and numbers in 40! ways which will also make the adversaries their job difficult. For messages with length more than $m$, split the message into blocks with length less than $m$.
\end{example}
\textbf{The Security of RSA \cite{ws}} \\

Five possible approaches to attacking the RSA algorithm are:\begin{itemize}
                                                              \item Brute force: This involves trying all possible private keys. To defend against this attack, use a large key space.
                                                              \item Mathematical attacks: There are several approaches, all equivalent in effort to factoring n into standard form. To overcome this threat take n as product of two large primes with at least 300 digits. 
                                                              \item Timing attacks: These depend on the running time of the decryption algorithm. They can be countered by constant exponentiation time, random delays and blinding.
                                                              \item Hardware fault-based attack: This involves inducing hardware faults in the processor that is generating random digital signatures. This is not a serious threat as it requires that the attacker have physical access to the target machine and that the attacker is able to directly control the input power to the processor .
                                                              \item Chosen ciphertext attacks: This type of attack exploits properties of the RSA algorithm. To overcome this simple attack, practical RSA-based crptosystems randomly pad the plaintext prior to encryption. 
                                                            \end{itemize}

\section{RSA-HGR Cryptosystem}

Here we are considering a cryptosystem based on modulo $n$. The following table gives numbers and the corresponding symbols.
\begin{center}
\begin{tabular}{|c|c|}

 \hline
  Symbols & $u$ values which are units\\
  & $\frac{\phi(n)!}{(\phi(n)-40)!}$ assignments uniquely\\
  \hline
  0 to 9 & $u_{1} $ to $ u_{10}$\\
  A to Z & $u_{11} $ to $ u_{36}$\\
  blank space & $u_{37}$ \\
   colon & $u_{38}$\\
   period & $u_{39}$\\
   hyphen &$u_{40}$\\
\hline
\end{tabular}
\end{center}
\noindent
Note: For $n=100, \phi(n)=40$. So there are $40!=815,915,283,247,897,734,345,611,269,596,115,894,272,000,000,000$ assignments of units to symbols.  \\
In this cryptosystem also, there are two stages. In stage 1, we shall compute the value of $\omega$ which Bob keeps secret and in stage 2, Alice shall decrypt the message sent by Bob. \\
In RSA, the challenge of adversaries is to find the value of $\phi(n)$. But here they have an extra challenge of locating or calculating the value of $\omega$. \\
\textbf{Stage 1-RSA}\\
Stage 1 is same as above. \\
\textbf{Stage 2-Halidon Group Ring (HGR)}\\
\textbf{Cryptosystem setup}
\begin{enumerate}
  \item Alice chooses programme 4 as the encryption key.
  \item Alice chooses programme 5 as the decryption key.
\end{enumerate}
\textbf{Cryptosystem Encryption}(Equation 1)
\begin{enumerate}
  \item Bob looks up Alice's encryption key.
  \item Bob writes his message $x=x_{1}x_{2}x_{3}\cdot \cdot \cdot x_{m}$.
  \item Bob translates $x$ into  $y=\lambda_{1}\lambda_{2}\lambda_{3}\cdot \cdot \cdot \lambda_{m}$ using the table.
  \item Bob calculates the coefficients of the corresponding unit using the programme 4 and the chosen value of $\omega$ in stage 1. 
  \item Bob sends coefficients to Alice.
\end{enumerate}
\textbf{Cryptosystem Decryption} (Theorem \ref{d1})
\begin{enumerate}
  \item Alice receives coefficients from Bob.
  \item Alice computes $\lambda_{1},\lambda_{2},\lambda_{3},\cdot \cdot \cdot ,\lambda_{m}$ using programme 5 with $\omega$ calculated in stage 1.
  \item Alice recovers the message using the table. 
\end{enumerate}

\begin{example}
\textbf{Stage 1}\\
Same as example 3.\\

\textbf{Stage 2} \\
Public Key \\
\begin{tabular}{|l|l|l|l|}
  \hline
  Symbols & Unit in $\mathbb{Z}_{491063}$& Symbols & Unit in $\mathbb{Z}_{491063}$ \\
   & assigned & & assigned \\
  \hline
  0 & 221373 & K & 80303 \\ \hline
  1 & 389086 & L & 52853 \\ \hline
  2 & 21415 & M & 80303 \\ \hline
  3 & 428230 & N & 52853 \\ \hline
  4 & 162920 & O & 114288 \\ \hline
  5 & 126345 & P & 473119 \\ \hline
  6 & 81308 & Q & 323343 \\ \hline
  7 & 490630 & R & 26857 \\ \hline
  8 & 22673 & S & 91043 \\ \hline
  9 & 4004 & T & 98057 \\ \hline
  A & 162483 & U & 150255 \\ \hline
  B & 2255 & V & 24495 \\ \hline
  C & 183775 & W & 86867 \\ \hline
  D & 4129 & X & 176089 \\ \hline
  E & 221927 & Y & 206140 \\ \hline
  F & 437699 & Z & 461772 \\ \hline
  G & 275130 & BLANK SPACE & 348362 \\ \hline
  H & 50473 & COLON & 90605 \\ \hline
  I & 123651 & PERIOD & 5932 \\ \hline
  J & 114773 & HYPHEN & 275062 \\ 
  \hline
\end{tabular}

Alice shared Bob the encryption key Discrete Fourier Transform and  $n=491063.$
Suppose Bob sends the following secret message to Alice. \\
\begin{center} AN IMMINENT ATTACK ON YOU WILL HAPPEN TOMORROW EVENING AT 5:30 PM. BE ALERT AND TAKE PRECAUTIONS.  \end{center}
The length of the message is 97 and to make it 202 we need to add 105 blank spaces. This can be translated into a 202 component vector.\\
 $ \left(
                                   \begin{array}{c}
 162483 \ 52853\ 348362 \ 123651 \ 80303\ 80303 \ 123651 \ 52853\ 221927 \ 52853 \\ 
 98057 \ 348362 \ 162483\ 98057 \ 98057\ 162483  \  183775 \ 80303 \ 348362\ 114288 \\ 
 52853 \ 348362 \ 206140 \ 114288 \ 150255  \ 348362 \ 86867 \ 123651 \ 52853 \ 52853 \\ 
 348362 \ 50473 \  162483 \ 473119 \ 473119 \ 221927 \ 52853 \ 348362 \ 98057\ 114288 \\ 
 80303 \ 114288 \ 26857 \ 26857 \ 114288 \ 86867 \ 348362 \  221927 \ 24495 \ 221927 \\ 
 52853\ 123651 \ 52853 \ 275130 \ 348362 \ 162483\ 98057\ 348362 \ 126345 \ 90605 \\ 
 428230 \ 221373 \ 348362 \ 473119 \ 80303 \ 5932 \ 348362 \ 2255 \ 221927 \ 348362 \\
  162483\ 52853 \ 221927\ 26857 \ 98057 \ 348362 \ 162483 \ 52853 \ 4129 \ 348362 \\ 
  98057 \ 162483\ 80303 \ 221927 \ 348362 \ 473119 \  26857 \ 221927 \ 183775 \ 162483 \\
  150255  \ 98057 \ 123651 \ 114288 \ 52853 \ 91043 \ 5932 \ 348362\ 348362 \ 348362 \\ 
  348362 \ 348362 \ 348362 \ 348362 \ 348362 \ 348362 \ 348362 \ 348362 \ 348362 \ 348362 \\ 
  348362 \ 348362 \ 348362 \ 348362 \ 348362 \ 348362 \ 348362 \ 348362 \ 348362 \ 348362 \\ 
  348362 \ 348362 \ 348362 \ 348362 \ 348362 \ 348362 \ 348362 \ 348362 \ 348362 \ 348362 \\ 
  348362 \ 348362 \ 348362 \ 348362 \ 348362 \ 348362 \ 348362 \ 348362 \ 348362 \ 348362 \\ 
  348362 \ 348362 \ 348362 \ 348362 \ 348362 \ 348362 \ 348362 \ 348362 \ 348362 \ 348362 \\ 
  348362 \ 348362 \ 348362 \ 348362 \ 348362 \ 348362 \ 348362 \ 348362 \ 348362 \ 348362 \\ 
  348362 \  348362 \ 348362 \ 348362 \ 348362 \ 348362 \ 348362 \ 348362 \ 348362 \ 348362 \\ 
  348362 \ 348362 \ 348362 \ 348362 \ 348362 \ 348362 \ 348362 \ 348362 \ 348362 \ 348362 \\
   348362 \ 348362 \ 348362 \ 348362 \ 348362 \ 348362 \ 348362 \ 348362 \ 348362 \ 348362 \\
   348362 \ 348362 \ 348362 \ 348362 \ 348362 \ 348362 \ 348362 \ 348362 \ 348362 \ 348362 \\ 
   348362 \ 348362 

                                  \end{array}
                                 \right)$
Bob chooses a primitive $202^{th}$ roots of unity $\omega$ in stage 1 using  programme 1 and kept secret in the halidon ring $\mathbb{Z}_{491063}$ with index 202. Using programme 4, Bob converts the plain text into the following cipher text in terms of coefficents. \\

 $ \left(
                                   \begin{array}{c}
 252493 \ 450589\ 460479 \ 204758 \ 233506\ 353306\ 421232 \ 356924 \ 301091 \ 289893 \\ 
 288179 \ 242097 \ 326234 \ 13515 \ 346524 \ 267905 \ 60544 \ 1589 \ 224877 \ 392891 \\ 
 393603 \ 346149 \ 126356 \ 374713 \ 42452 \ 30660 \ 444474 \ 328107 \ 278316 \ 320329 \\ 
 215968 \ 8062 \ 69501 \ 442389 \ 463363 \ 20437 \ 184879 \ 111644 \ 215157 \ 487962 \\ 
 182507 \ 157039 \ 200299 \ 355976 \ 90232 \ 362884 \ 407252 \ 282817 \ 324527 \ 299628 \\ 
 83392 \ 380613 \ 274931 \ 455342 \ 28745 \ 445319 \ 430230 \ 446985 \ 347595 \ 201469 \\
 91852 \ 53863 \ 48802 \ 172649 \ 95573 \ 70434 \ 71251 \ 95329 \ 257149 \ 125640 \\
  436246 \ 37716 \ 452002 \ 143402 \ 221576 \ 137122 \ 379802 \ 91038 \ 217808 \ 73515 \\ 
  245279 \ 62765 \ 16846 \ 473375 \ 284904 \ 470346 \ 392515 \ 31311 \ 386722 \ 228015 \\
  471883 \ 95686 \ 284880 \ 373228 \ 282251 \ 461945 \ 347587 \ 372751 \ 243942 \ 339087 \\ 
  441737 \ 321411 \ 205845 \ 172853 \ 450407 \ 431493 \ 72268 \ 378074 \ 403244 \ 261526 \\ 
  363362 \ 372773 \ 193094 \ 61896 \ 76335 \ 442360 \ 12418 \ 333213 \ 349588 \ 137997 \\ 
  465244 \ 464347 \ 453371 \ 370624 \ 414389 \ 329819 \ 99661 \ 168143 \ 270109 \ 194801 \\ 
  460848 \ 483049 \ 98372 \ 225436 \ 184156 \ 147000 \ 137130 \ 254978 \ 435708 \ 227589 \\ 
  126220 \ 45283 \ 312941 \ 108458 \ 176782 \ 55396 \ 134718 \ 440134 \ 367637 \ 450466 \\ 
  32149 \ 44665 \ 445959 \ 120765 \ 447216 \ 362999 \ 402427 \ 210408 \ 171884 \ 486885 \\ 
  280531 \ 322673 \ 116715 \ 483483 \ 398994 \ 31300 \ 134031 \ 431195 \ 434524 \ 172474 \\ 
  198368 \ 111628 \ 469394 \ 198059 \ 11214 \ 387413 \ 93105 \ 390274 \ 263412 \ 304750 \\
   333166 \ 415475 \ 31915 \ 125737 \ 36184 \ 115899 \ 390465 \ 6472 \ 173688 \ 208819 \\
   168514 \ 197636 \ 136348 \ 410545 \ 200343 \ 316617 \ 47292 \ 286043 \ 112122 \ 239726 \\ 
   361815 \ 85601

                                   \end{array}
                                 \right)$
 
\end{example}
Alice receives the above cipher text and  she uses  $\omega=239823$ from stage 1. Applying the  programme 5 Alice gets the original message back.
\section{Conclusion}\label{sec13}
These new cryptosystems have been developed using halidon rings, halidon group rings and and Discrete Fourier Transforms. These systems provides high-level security for communication between ordinary people or classified messages in government agencies. The level of security can be increased by utilising advanced computer technology and powerful codes to calculate the primitive $m^{th}$ root of unity for a very large value of $n$ where the calculation of $\phi(n)$ is difficult. There are scopes for the development of new cryptosystems based on Cyclotomic polynomials..

\end{document}